\documentclass[submission]{eptcs}
\providecommand{\event}{DICE 2010} % Name of the event you are submitting to
\providecommand{\volume}{??}
\providecommand{\anno}{2010}
\providecommand{\firstpage}{1}
\providecommand{\eid}{??}
\usepackage{breakurl}        % Not needed if you use pdflatex only.

\usepackage{stmaryrd}
\usepackage[mathletters]{ucs}
\usepackage[utf8x]{inputenc}

\usepackage{amsmath}
\usepackage{amsthm}
\usepackage{amssymb}

\usepackage{bussproofs}
\usepackage{color}

\definecolor{mygray}{gray}{0.35}
\newcommand{\faded}[1]{\textcolor{mygray}{#1}~}

\DeclareUnicodeCharacter{8871}{\models}
\DeclareUnicodeCharacter{8614}{\mapsto}
\DeclareUnicodeCharacter{8499}{\mathcal{M}}
\DeclareUnicodeCharacter{8873}{\Vdash}
\DeclareUnicodeCharacter{10214}{\llbracket}
\DeclareUnicodeCharacter{10215}{\rrbracket}

\theoremstyle{plain}
\newtheorem{thm}{Theorem}
\theoremstyle{definition}
\newtheorem{defn}[thm]{Definition}
\theoremstyle{definition}
\newtheorem{example}[thm]{Example}
\theoremstyle{remark}
\newtheorem*{rem*}{Remark}
\theoremstyle{plain}
\newtheorem{lem}[thm]{Lemma}
\theoremstyle{plain}

\theoremstyle{plain}

\global\long\def\Prop{Prop}

\global\long\def\Type{Type}

\global\long\def\one{\texttt{unit}}

\global\long\def\nat{\texttt{nat}}

\global\long\def\V{{V}}
\global\long\def\H{\mathcal{H}}
\global\long\def\Sign{Σ}
\global\long\def\plus{plus}
\global\long\def\elem{\mathcal{E}}
\global\long\def\sumb{sum}

\global\long\def\mult{mult}
\global\long\def\succe{s}
\global\long\def\minus{minus}
\global\long\def\pred{pred}

\global\long\def\prodb{prod}
\global\long\def\ok{⊢_{\!\text{\tiny ok}}}
\global\long\def\ell{⊢_{\!\text{\tiny eal}}}

\global\long\def\Pgr{\mathcal{P}}

\newcommand{\reduce}[1]{\xrightarrow{\text{#1}}}

\newsavebox{\fmbox}
\newenvironment{fmpage}[1]
     {\begin{center}\begin{lrbox}{\fmbox}\begin{minipage}{#1}}
     {\end{minipage}\end{lrbox}\fbox{\usebox{\fmbox}}\end{center}}

\title{Controlling program extraction in Elementary Linear Logic}
\author{Marc Lasson
\institute{École Normale Supérieure de Lyon, France}
\email{marc.lasson@ens-lyon.org}
}
\def\titlerunning{Elementary second-order functionnal arithmetic}
\def\authorrunning{Marc Lasson}
\begin{document}
\maketitle

\begin{abstract}

We present an adaptation, based on program extraction in elementary linear
logic, of Krivine \& Leivant's system FA₂.  This system allows to write
higher-order equations in order to specify the computational content of
extracted programs.  
The user can then prove a generic formula, using these equations as axioms,
whose proof can be extracted into programs that normalize in elementary time and
satisfy the specifications.
Finally, we show that every elementary
recursive functions can be implemented in this system. 

\end{abstract}

\section*{Introduction}

Elementary linear logic is a variant of linear logic introduced by
Jean-Yves Girard in an appendix of~\cite{LLL} that characterizes, through the
Curry-Howard correspondence, the class of elementary recursive
functions.
There are two usual ways to program in such a light logic: by using it
as a type system of a $λ$-calculus or by extracting programs from
proofs in a sequent calculus (see \cite{DanosJoinet} for instance).

The former is used for propositional fragments of Elementary Affine Logic
in~\cite{Simone} and of Light Affine Logic in~\cite{Baillot}.  However, when the
pogrammer provides a $λ$-term which is not typable, he has no
clue to find a suitable term implementing the same function.
In the later approach, the programmer must keep in mind the underlying
computational behaviour of his function during the proof and check
later, by external arguments, that the extracted $λ$-term
implements the desired function.

In this paper, we describe a system in which we try to make the second
approach a bit more practical.
Firstly because our system is endowed with a kind of proof irrelevance: all
proofs of the same formula are extracted to extensionally equivalent
terms; and then because the program automatically satisfy the given
specification used as axioms during the proof.

FA₂ is an intuitionistic second-order logic whose formulas are built upon
first-order terms, predicate variables, arrows and two kind of quantifiers, one on
first-order variables and the other on predicate variables. Jean-Louis Krivine
described in \cite{LambdaCalculusTypesAndModels} a methodology to use
this system for programming with proofs. In this system, the induction principle 
for integers may be expressed by 
  $$∀X, (∀y,X\,y ⇒ X\,(s\,y)) ⇒ X\,0 ⇒ X\,x.$$
This formula is written $N\,x$ and it is used to represent integers. The
programmer then gives some specifications of a function. For instance for the
addition, he may give: 
\begin{eqnarray*}
  \plus(0, y) & = & y \\
  \plus (s(x), y) &  = & s(\plus(x,y)).
\end{eqnarray*}
Now, if he finds a proof of $$∀x\,y,N\,x ⇒ N\,y ⇒ N\,(\plus(x,y))$$ in which he
is allowed to rewrite formulas with the specifications, then it is proved that
the $λ$-term extracted from this proof using standard techniques is a program
satisfying the specifications.

\noindent
We have adapted the system FA₂ of Leivant and Krivine following two directions:
\begin{itemize}
  \item We replace the grammar of first-order terms by the   
  whole $λ$-calculus. We can then extract higher-order functions
  instead of purely arithmetical functions. We have shown in~\cite{AF2PTS} 
  that the resulting system can be described as a pure type system (PTS). We
  have also built an extensionnal model, and re-adapted realizability tools
  for it. Here we only present the material needed for elementary
  programming and we refer the reader to~\cite{AF2PTS} for more details.
  \item We ensure complexity bounds by making its logic elementary.
\end{itemize}

In the next section, we introduce the grammar for our formulas and describe how
we interpret them. In section 2, we present our proof system and how we can
program with it. In the last section, we prove that we characterize the class of
elementary recursive functions. We bring our system back to the usual Elementary
Affine Logic in order to have the correctness. Finally we give two proofs of the
completeness: one by using the completeness of \emph{Elementary Affine Logic} 
(henceforth EAL) and the other by invoking,
like in~\cite{DanosJoinet}, Kalmar's characterization of elementary functions.
We present the second proof as an illustration of how to program in our system.
Indeed, it will give the programmer a direct way to code elementary functions
without having to encode them in EAL.

\section{Types, First-Order Terms and Formulas}
We assume for the rest of this document that we have at our disposal
three disjoint sets of infinitely many variables:
\begin{itemize}
  \item the set $\V_0$ of so-called {\it type variables} whose elements are
  denoted with letters from the beginning of the Greek alphabet and some 
  variations around them (ie. $α$, $β$, $α₁$, $α₂$, ...), 
  \item the set $\V_1$ of {\it first-order variables} whose elements are denoted with 
  letters from the end of the Latin alphabet (ie. $x$, $y$, $z$, $x₁$, $x₂$, ...),
  \item the set $\V_2$ of {\it second-order variables} whose elements are denoted with 
  uppercase letters from the end of the Latin alphabet (ie. $X$, $Y$, $Z$, $X₁$, $X₂$, ...).
\end{itemize}
We also assume that we have an injection of second-order variables into type variables
and write $α_X$ the image of a variable $X$ by this injection. This will be useful later 
when we will send formulas onto system $\mathcal{F}$ types by a forgetful projection.

\begin{defn}
The following grammars define the terms of the system:
\begin{enumerate}
\item Types are system $\mathcal{F}$ types:
\vspace{-0.5em}
\begin{equation*}
τ,σ,...  \quad := \quad α
  \quad|\quad ∀α,τ
  \quad|\quad σ → τ\end{equation*}
\vspace{-2.5em}
\item First-order terms are Church-style $\lambda$-calculus terms:
\vspace{-0.5em}
\begin{equation*}
s,t,... \quad := \quad x
 \quad | \quad (s\, t)
 \quad | \quad (t\,τ)
 \quad | \quad \lambda x:τ.t
 \quad | \quad \Lambda α.t\end{equation*}
\vspace{-2.5em}
\item Finally, second-order formulas are given by the following grammar:
\vspace{-0.5em}
\begin{equation*}
P,Q,... \quad := \quad X\, t_{1}\, t_{2}\,...\, t_{n}\\
 \quad | \quad P ⊸ Q\\
 \quad | \quad ∀ X:[τ_{1},..,τ_{n}],\, P\\
 \quad | \quad ∀ x:τ,\, P\\
 \quad | \quad ∀ α,\, P
 \quad | \quad ! P\end{equation*} 
\end{enumerate}
\end{defn}

Theses grammars describe terms that will be used in this paper, 
$λ$, $Λ$ and the three different $∀$ behave as binders like in usual
calculi. We always consider terms up to $α$-equivalence and we do
not bother with capture problems. We also admit we have 
six notions of substitution which we assume to be well-behaved with regard to 
the $α$-equivalence (all these notions are more seriously defined in \cite{AF2PTS}): 
\begin{enumerate}
  \item the substitution $τ[σ/α]$ of a type variable $α$ by a type $σ$ in a type $τ$, 
  \item the substitution $t[τ/α]$ of a type variable $α$ by a type $τ$ in a first-order term $t$, 
  \item the substitution $t[s/x]$ of a first-order variable $x$ by a first-order term $s$ in a first-order term $t$, 
  \item the substitution $P[τ/α]$ of a type variable $α$ by a type $τ$ in a formula $P$, 
  \item the substitution $P[t/x]$ of a first-order variable $x$ by a first-order term $t$ in a formula $P$, 
  \item the substitution $P[Q/X\,x_1\,...\,x_n]$ of a second-order variable $X$ by a formula $Q$ with parameters $x₁,...,x_n$ in a formula $P$. 
\end{enumerate}

  The last one is not very usual (the notation comes from
\cite{LambdaCalculusTypesAndModels}): it replaces occurrences of the form
$X\,t_1\,...t_n$ by the formula $Q[t_1/x_n]...[t_n/x_n]$ and it is not defined
if $P$ contains occurrences of $X$ of the form $X\,t_1\,...\,t_k$ with $k \not=
n$. The simple type system we are going to define will guarantee us that such
occurrences cannot appear in a well-typed formula.

And since we can build redexes in terms (of the form $((λx:τ.t₁)\,t₂)$ and
$((Λα.t)\,τ)$) we have a natural notion of $β$-reduction for first-order terms
which we can extend to formulas (we write $t₁ >_β t₂$ and $P₁ >_β P₂$ for the
transitive closure of the $β$-reduction on first-order terms and formulas). 

We adopt the usual conventions about balancing of parentheses: arrows are
right associative (it means that we write $A ⊸ B ⊸ C$ instead of $A ⊸ (B ⊸ C)$) 
and application is left associative (meaning we write $t₁\,t₂\,t₃$ instead of 
$(t₁\,t₂)\,t₃$). By abuse of notation, we allow ourselves not to write the type
of first and second order $∀$ when we can guess them from the context. We 
also write $!^k P$ instead of $!...!P$ with $k$ exclamation marks. 

\begin{example}
  Here are some examples of formulas of interest : 
\begin{enumerate}
\item Leibniz's equality between two terms $t₁$ and $t₂$ of type $τ$ 
\[ \forall X:[\tau],X\, t₁⊸ X\, t₂\]
which we write it $t₁ =_τ t₂$ in the remaining of this document. 
\item The induction principle for a natural number $x$ \[
\forall X:[\nat],
  !(\forall y,X\, y⊸ X\,(s\, y)) ⊸ 
  !(X\,0⊸ X\, x)\]
 which we write $N\, x$ where $\nat$ will be the type $
\forall\alpha,(\alpha\rightarrow\alpha)\rightarrow\alpha\rightarrow\alpha$
of natural numbers in system $\mathcal{F}$ and where $s$ and $0$ are first-order variables.
\item The tensor between two formulas $P$ and $Q$,
   $\forall X,(P⊸ Q⊸ X)⊸ X$ written $P ⊗ Q$. 
\item And the extensionality principle  \[
\forall\alpha\,\beta,\forall f\, g:\alpha\rightarrow\beta,(\forall x:\alpha,f\, x=_{\beta}g\, x)⊸ f=_{\alpha\rightarrow\beta}g\]
\end{enumerate}
\end{example}

\begin{defn}\label{DefContext}
 A \textit{context} is an ordered list of elements of the form: 
\begin{equation*}
  \begin{array}{lcccr}
     α : \Type &\text{ or }& x : τ &\text{ or }& X : [τ_1, ..., τ_n]. \\
  \end{array}
\end{equation*}
In the following, the beginning of the lowercase Latin alphabet $a, b, ...$ 
will designate variables of any sort and the beginning of uppercase Latin
alphabet $A, B, C, ...$ designate $\Type$, $\Prop$, any type $τ$ or something
of the form $[τ_1, ..., τ_n]$.
We write $a ∈ Γ$, if there is an element of the form $a : \_$ in
$Γ$. A context $Γ$ is said to be \textit{well-formed} if ``$Γ$ is well-formed" can be derived in the
type system. A formula $F$ (resp. a term $t$, resp. a type $τ$) is said to be
well-formed in a context $Γ$ if the sequent $Γ \ok F : \Prop$ (resp.
$Γ \ok t : τ$ for some $τ$, resp. $Γ \ok τ : \Type$) is derivable in the type
system. 
\end{defn}

\begin{fmpage}{0.95\textwidth}
\begin{prooftree}
  \AxiomC{}
  \UnaryInfC{$\emptyset$ is well-formed}

  \AxiomC{$Γ$ is well-formed} 
  \RightLabel{$α \not∈ Γ$}
  \UnaryInfC{$Γ, α:\Type$ is well-formed} 

\noLine
\BinaryInfC{}
\end{prooftree}
\begin{prooftree}

  \AxiomC{$Γ ⊢ τ : \Type$}
  \LeftLabel{$x \not∈ Γ$}
  \UnaryInfC{$Γ, x:τ$ is well-formed} 

  \AxiomC{$Γ⊢τ_1:\Type$ ... $Γ⊢τ_n:\Type$}
  \RightLabel{$X \not∈ Γ$}
  \UnaryInfC{$Γ, X:[τ_1,...,τ_n]$ is well-formed}
  
\noLine
\BinaryInfC{}
\end{prooftree}
\begin{prooftree}
  \AxiomC{$Γ$ is well-formed} 
  \UnaryInfC{$Γ,a:A \ok a:A$}
  
  \AxiomC{$Γ \ok b : B$}
  \RightLabel{$a \not= b$}
  \UnaryInfC{$Γ,a:A \ok b : B$}

  \AxiomC{$Γ \ok P : \Prop$}
  \UnaryInfC{$Γ \ok !P : \Prop$}
\noLine
\TrinaryInfC{}
\end{prooftree}\begin{prooftree}

\AxiomC{$Γ \ok τ : \Type$}
\AxiomC{$Γ \ok σ : \Type$}
\BinaryInfC{$Γ \ok τ → σ : \Type$}

\AxiomC{$Γ, α:\Type \ok τ : \Type$}
\UnaryInfC{$Γ \ok (∀ α, τ) : \Type$}

\AxiomC{$Γ, x:τ \ok t : σ$}
\UnaryInfC{$Γ \ok (\lambda x:τ.t) : τ → σ$}

\noLine
\TrinaryInfC{}
\end{prooftree}\begin{prooftree}
\AxiomC{$Γ, α : \Type \ok t : τ$}
\UnaryInfC{$Γ \ok (\Lambda α.t) : ∀ α,τ$}
\AxiomC{$Γ \ok f : τ → σ$}
\AxiomC{$Γ \ok a : τ$}
\BinaryInfC{$Γ \ok (f\,a):σ$}

\AxiomC{$Γ \ok f : \Lambdaα.σ$}
\AxiomC{$Γ \ok τ : \Type$}
\BinaryInfC{$Γ \ok (f\,τ):σ[τ/α]$}

\noLine
\TrinaryInfC{}
\end{prooftree}\begin{prooftree}
\AxiomC{\hspace{-1em}$Γ, X:[τ_1,...,τ_n] \ok Q : \Prop$}
\UnaryInfC{$Γ \ok (∀ X : [τ_1,...,τ_n], Q) : \Prop$}

\AxiomC{$Γ, x:τ \ok Q : \Prop$}
\UnaryInfC{$Γ \ok (∀ x : τ, Q) : \Prop$}

\AxiomC{$Γ, α:\Type \ok Q : \Prop$}
\UnaryInfC{$Γ \ok (∀ α, Q) : \Prop$}

\noLine
\TrinaryInfC{}

\end{prooftree}\begin{prooftree}
\AxiomC{$Γ \ok P : \Prop$}
\AxiomC{$Γ \ok Q : \Prop$}
\BinaryInfC{$Γ \ok (P ⊸ Q) : \Prop$}

\AxiomC{$Γ \ok t_1 : τ_1 \quad \cdots \quad Γ \ok t_n : τ_n$}
\AxiomC{$Γ \ok X : [τ_1,...,τ_n]$}
\BinaryInfC{$Γ \ok X\,t_1\,...\,t_n : \Prop$}

\noLine
\BinaryInfC{}
\end{prooftree}
\begin{center}
\textbf{Type system for checking well-formedness}
\end{center}
\end{fmpage}

\begin{example}
  These formulas are well-typed :
  \begin{enumerate}
    \item $Γ,x:τ,y:τ \ok x =_τ y : \Prop$, 
    \item $Γ,\succe:\nat⊸\nat,0:\nat,x:\nat \ok N x : \Prop$, 
    \item $Γ,X:\Prop,Y:\Prop \ok X ⊗ Y : \Prop$, 
    \item $\ok \forall\alpha\,\beta,\forall f\, g:\alpha\rightarrow\beta,(\forall
x:\alpha,f\, x=_{\beta}g\, x)⊸ f=_{\alpha\rightarrow\beta}g : \Prop$.
  \end{enumerate}
\end{example}

~\\
We have shown in \cite{AF2PTS} that this simple system have numerous good
properties of pure type systems (like subject reduction).

\subsection*{Interpretations in standard models}

In this section, we build a small realizability model for our proof system which
we will use later to prove the correctness with respect to the specification 
of the extracted proof. One of our goal is to make the model satisfy the
extensionality principle, because we will need to be able to replace in our
proofs higher-order terms by other extensionally equal terms. 

We define the set $\Pgr$ of programs to be the set of pure $λ$-terms modulo
$β$-reduction.  In the following, we interpret terms in $\Pgr$, types by partial
equivalence relations (PER) on $\Pgr$ and second-order variables by sets of
element in $\Pgr$ stable by extensionality (you are not allowed to consider sets
which are able to distinguish terms that compute the same things). Finally,
formulas are interpreted as classical formulas: all informations about
linearity and exponentials are forgotten. Indeed, we forget all complexity
informations because the only purpose of model theory here is to have result
about the compliance with respect to the specifications. 

\begin{defn}
Let $Γ$ be a well-formed context. A $Γ$-model consists of three partial
functions recursively define below. The first one is map from type variables 
to PERs, the second is a map from first-order variables to $\Pgr$ and the last
one is a map from second-order variables to sets of tuples of programs. 
\begin{itemize}
\item If $Γ$ is empty, then the only $Γ$-model is three empty maps. 
\item If $Γ$ has the form $Δ, x : τ$ and if $ℳ=(ℳ_0, ℳ_1,ℳ_2)$ 
is a $Δ$-model, then for any $t ∈ ⟦τ⟧_ℳ$, $(ℳ_0, ℳ_1[x↦t],ℳ_2)$ is a $Γ$-model
(in the following, we simply write it $ℳ[x↦t]$).
\item If $Γ$ has the form $Δ, α : \Type$ and if $ℳ=(ℳ_0, ℳ_1,ℳ_2)$
  is a $Δ$-model, then for any PER $R$, 
  $(ℳ_0[α↦ R], ℳ_1,ℳ_2)$ is a $Γ$-model (we write it $ℳ[α↦R]$).
\item If $Γ$ has the form $Δ, X : [τ_1, ...,τ_n]$ and if  $ℳ=(ℳ_0, ℳ_1,ℳ_2)$
  is a $Δ$-model, then for any $E ⊆ ⟦τ_1⟧_ℳ \times...\times⟦τ_n⟧_ℳ$ such
  that $E$ satisfy the \emph{stability condition}
   $$\text{If }(t₁,...,t_n) ∈ E ∧ t₁ ∼^ℳ_{τ₁} t₁' ∧ ... ∧ t_n ∼^ℳ_{τ_n} t_n'\text{, then}
     (t'₁,...,t'_n) ∈ E$$ 
  $(ℳ_0, ℳ_1,ℳ_2[X↦ E])$ is a $Γ$-model (we write it $ℳ[X↦E]$).
\end{itemize}
Where $∼^ℳ_τ$ is a partial equivalence relation whose domain is written $⟦τ⟧_ℳ$
defined recursively on the structure of $τ$, 
  \begin{itemize}
    \item $∼^ℳ_α$ is equal to $ℳ₀(α)$, 
    \item $∼^ℳ_{σ→τ}$ is defined by 
      $t₁ ∼^ℳ_{σ→τ} t₂ ⇔ ∀ s₁ s₂, s₁ ∼^ℳ_σ s₂ ⇒ (t₁\,s₁) ∼^ℳ_τ (t₂\,s₂),$ 
    \item $∼^ℳ_{∀α,τ} = \bigcap_{R\text{ is PER}} ∼^{ℳ[α↦R]}_τ$.
  \end{itemize}
Intuitively $t₁ ∼^ℳ_τ t₂$ means the pure $λ$-terms $t₁$ and $t₂$ are of type $τ$
and they are extensionally equivalent. 
\end{defn}

Now, we can define the interpretation $⟦t⟧_ℳ$ of a first-order term $t$ such that $Γ \ok t
: τ$ in a $Γ$-model $ℳ$ to be the pure $λ$-term obtained by replacing all occurrences
of free variables by their interpretation in $ℳ$ and by erasing type
information. And we can prove substitution lemmas. 

\begin{lem} For any $Γ$-models $ℳ$, 
\begin{enumerate} 
  \item If $Γ, α:\Type \ok τ:\Type$ and $Γ \ok σ :\Type$, then
     $⟦τ[\sigma/α]⟧_ℳ = ⟦τ⟧_{ℳ[α↦∼^ℳ_σ]}$,
  \item If $Γ, α : \Type \ok t : σ$ and $Γ \ok τ : τ$, then 
     $⟦t[τ/α]⟧_ℳ = ⟦t⟧_{ℳ[α↦∼^ℳ_τ]}$
  \item If $Γ, x : σ \ok t : τ$ and $Γ \ok s : σ$, then 
     $⟦t[s/x]⟧_ℳ = ⟦t⟧_{ℳ[x↦⟦s⟧_ℳ]}$
  \item If $Γ \ok t : τ$, $t \equiv_{β} t'$ and $Γ \ok t':τ$, then $⟦t⟧_ℳ = ⟦t'⟧_ℳ$. 
\end{enumerate}
\end{lem}

\noindent And then we can deduce an adequacy lemma about well-typed terms. 

\begin{lem}
  If we have $Γ \ok t : τ$ and $ℳ$ a $Γ$-model, then $⟦t⟧_ℳ ∈ ⟦τ⟧_ℳ$.
\end{lem}

Now we can define the notion of satisfiability in a model recursively on 
formulas' structure. 
\begin{defn}
Let $P$ be a formula such that $Γ \ok P:\Prop$ and $ℳ$ be a
$Γ$-model.

\begin{itemize}
  \item $ℳ ⊧ X\,t_1\,...\,t_n$ iff $(⟦t_1⟧_ℳ, ...,⟦t_n⟧_ℳ) ∈ ℳ(X)$, 
  \item $ℳ ⊧ P ⊸ Q$ iff $ℳ ⊧ P$ implies $ℳ ⊧ Q$, 
  \item $ℳ ⊧ ∀ X:[τ_1, ...,τ_n], P$ iff for all
     $E ⊆ ⟦τ_1⟧_ℳ \times...\times⟦τ_n⟧_ℳ$ satisfying the stability 
     condition, $ℳ[X ↦ E]⊧ P$, 
  \item $ℳ ⊧ ∀ x:τ, P$ iff for all $t ∈ ⟦τ⟧_M$, 
     $ℳ[x↦ t] ⊧ P$,
  \item $ℳ ⊧ ∀ α, P$ iff for all PER $R$ on $\Pgr$, 
     $ℳ[α ↦ R] ⊧ P$,
  \item $ℳ ⊧ ! P$ iff $ℳ  ⊧ P$.
\end{itemize}

If $E$ is a set of formulas well-formed in $Γ$, for all $Γ$-model $ℳ $, we 
write $ℳ ⊧ E$ for meaning that $ℳ ⊧ Q$ for all $Q ∈ E$. And if $T$ is another
set of formulas well-formed in $Γ$, we write $T ⊧_Γ E$ if for all
$Γ$-model $ℳ$, $ℳ ⊧T$ implies $ℳ ⊧ E$ (and we write $T ⊧_Γ P$ in
place of $T ⊧_Γ \{P\}$).  
\end{defn}

\begin{lem}
The formulas are unable to distinguish extensionally equivalent 
programs: for any formula $P$ such that $Γ, x₁:τ₁, ...,x_n:τ_n \ok P : \Prop$
and any $Γ$-model $ℳ$ the set 
$$\{ (t_1,...,t_n) ∈ ⟦τ₁⟧×...×⟦τ_n⟧|
      ℳ[x_1↦ t_1, ...,x_n↦ t_n]⊧ P\}$$
satisfies the stability condition. 
\end{lem}

\begin{lem} For any $Γ$-model $ℳ$, 
\begin{enumerate} 
  \item 
    If $Γ, α : \Type \ok P : Prop$ and $Γ \ok τ:\Type$, 
    then $ℳ ⊧ P[τ/α] ⇔ ℳ [α ↦ ∼^ℳ_τ] ⊧ P$,
  \item 
    If $Γ, x : τ \ok P : Prop$ and $Γ \ok t:τ$, then  
    $ℳ ⊧ P[t/x] ⇔ ℳ [x ↦ ⟦t⟧_ℳ] ⊧ P$,
  \item 
    If $Γ, X:[τ₁,...,τ_n] \ok P : Prop$ and $Γ, x₁:τ₁, ...,x_n:τ_n\ok Q:\Prop$,
    then $$ℳ ⊧ P[Q/X\,x_1\,...,x_n] ⇔ ℳ[X↦ E]⊧ P$$ where 
  $$E = \{ (t_1,...,t_n) ∈ ⟦τ₁⟧×...×⟦τ_n⟧ |
      ℳ[x_1↦ t_1, ...,x_n↦ t_n]⊧ Q\},$$
  \item If $Γ \ok P : \Prop$, $P \equiv_{β} P'$ and $Γ \ok P' : \Prop$, then $ℳ ⊧ P ⇔ ℳ ⊧ P'$.
\end{enumerate}
\end{lem}

\begin{lem}
 If $Γ \ok t₁ : τ$, $Γ \ok t₂:τ$ and $ℳ$ is a $Γ$-model, then 
          $ℳ ⊧ t₁ =_τ t₂$ ⇔ $⟦t₁⟧_ℳ ∼^ℳ_τ ⟦t₂⟧_ℳ $.
\end{lem}
\begin{proof}~\\\vspace{-1.3em}
  \begin{itemize}
      \item $ℳ ⊧ t₁ =_τ t₂ ⇒ ⟦t₁⟧_ℳ ∼^ℳ_τ ⟦t₂⟧_ℳ $ : 
        Let $E = \{ t ∈ ⟦τ⟧_ℳ; ⟦t₁⟧_ℳ ∼^ℳ_τ t \}$ be the equivalence class of $⟦t₁⟧_ℳ$ 
        (as such $E$ satisfy the stability condition). 
        If $ℳ ⊧ t₁ =_τ t₂$, then $ℳ[X ↦ E]⊧ X\,t₁ ⊸ X\,t₂$ which means that 
        $⟦t₁⟧_ℳ ∈E$ -which is true- implies $⟦t₂⟧_ℳ∈E$ which means that $⟦t₁⟧_ℳ∼^ℳ_τ⟦t₂⟧_ℳ$. 
      \item  $⟦t₁⟧_ℳ ∼^ℳ_τ ⟦t₂⟧_ℳ ⇒ ℳ ⊧ t₁ =_τ t₂$ :
        Suppose $⟦t₁⟧_ℳ ∼^ℳ_τ ⟦t₂⟧_ℳ$, then for all $E ⊆ ⟦τ⟧_ℳ$ satisfying the stability condition, 
        we have $⟦t₁⟧_ℳ  ∈ E$ implies $⟦t₂⟧_ℳ ∈ E$ or in other words $ℳ[X ↦ E] ⊧ X\,t₁ ⊸ X\,t₂$.
        And therefore, we obtain $ℳ ⊧ t₁ =_τ t₂$. 
  \end{itemize}
\end{proof}

\begin{defn}
Suppose we have $Γ \ok P₁ : \Prop$, $Γ \ok t₁ : τ$ and $Γ \ok t₂ : τ$, we say that
  $P₁ \reduce{$t₁ = t₂$} P₂$ if there exists a formula $Q$
such that $Γ, x:τ \ok Q : \Prop$, $P₁ \equiv Q[t₁/x]$ and 
  $P₂ \equiv Q[t₂/x]$. 
\end{defn}

\begin{lem} ~\\\indent
 If $ℳ ⊧ t₁ =_τ t₂$ and $P₁ \reduce{$t₁ = t₂$} P₂$ then 
  $ℳ ⊧ P₁ ⇒ ℳ ⊧ P₂$.
\end{lem}
\begin{proof}
  Suppose $P₁ \equiv Q[t₁/x]$ and $P₂ \equiv Q[t₂/x]$. 
  Let $E$ be the set $\{t ∈ ⟦τ⟧_ℳ | ℳ[x↦t]⊧Q\}$. 
  Since $ℳ ⊧ t₁ =_τ t₂$, we have that $ℳ[X↦E]⊧X\,t₁⊸X\,t₂$
  which is equivalent to $ℳ[x↦⟦t₁⟧_ℳ] ⊧ Q$ implies $ℳ[x↦⟦t₂⟧_ℳ] ⊧ Q$,
  or $ℳ ⊧ P₁$ implies $ℳ ⊧ P₂$, or $ℳ ⊧P₁ ⊸ P₂$.  
\end{proof}

\begin{thm}~\\\indent
    Theses models satisfy the extensionality principle : 
  $$ℳ  ⊧ \forall\alpha\,\beta,\forall f\, g:\alpha\rightarrow\beta,(\forall
x:\alpha,f\, x=_{\beta}g\, x)⊸ f=_{\alpha\rightarrow\beta}g.$$
\end{thm}
\begin{proof}
  It is a consequence of the last two lemmas. 
  \begin{itemize}
    \item The last one gives us that $$ℳ ⊧∀αβ,\forall f\, g:\alpha\rightarrow\beta,(\forall
x:\alpha,f\, x=_{\beta}g\, x)⊸ (\forall
x\,y:\alpha, x =_α y ⊸ f\, x=_{\beta}g\, y).$$
    \item Therefore we are left to prove that 
$ℳ ⊧∀αβ,∀ f\, g:α→β,(∀ x\,y:α, x =_α y ⊸ f\, x=_{\beta}g\, y) ⊸ f=_{α⊸β}g$. 
  Let $R_α$ and $R_β$ be two PER, $t₁, t₂ ∈ ⟦α→β⟧_{ℳ[α↦R_α, β↦R_β]}$. 
  Suppose $ℳ[α↦R_α, β↦R_β,f↦t₁,g↦t₂] ⊧∀ x\,y:α, x =_α y ⊸ f\, x=_{\beta}g\, y$, 
  we need to prove that, $ℳ[α↦R_α, β↦R_β,f↦t₁,g↦t₂] ⊧ f=_{α→β}g$ or equivalently
  that $t₁∼^{ℳ [α↦R_α, β↦R_β]}_{α→β}t₂$, which is also equivalent to the fact that
  for all $(a₁, a₂) ∈ R_α$, $((t₁\,\,a₁), (t₂\,\,a₂)) ∈ R_β$ which is exactly
  $ℳ[α↦R_α, β↦R_β,f↦t₁,g↦t₂] ⊧∀ x\,y:α, x =_α y ⊸ f\, x=_{\beta}g\, y$.
  \end{itemize}
\end{proof}

\subsection*{Projecting formulas toward types}

In order to write the rules of our proof system in the next section, we
are going need to have way to project second-order formulas toward types. 

\begin{defn}
Given a formula $F$, we define the type $F^-$ recursively 
built from $F$ in the following way.
\begin{equation*}
\begin{array}{lccccr}
(X\, t_1\,...\, t_n)^- \equiv α_X &
(A ⊸ B)^- \equiv A^- → B^- &
(∀α,F)^- \equiv F^- &
(∀ x:α,F)^- \equiv F^-  &
(!F)^- \equiv !F^-  \\
\end{array}
\end{equation*}
\begin{equation*}
(∀ X:[τ_1,...,τ_n],F)^- \equiv ∀α_X,F^-.
\end{equation*}
\end{defn}

\begin{lem}\label{lemma_moins}
  If $Γ \ok A:\Prop$, $Γ^\star \ok A^-:\Type$ where $Γ^\star$ is obtained
  from $Γ$ by replacing occurrences of 
      ``$X:[τ_1,...,τ_n]$" by ``$α_X : \Type$" and letting others unchanged.
\end{lem}
\begin{example}~\\\vspace{-1.3em}
  \begin{itemize} 
     \item $(t₁ =_τ t₂)^- \equiv ∀α,α ⊸ α \equiv \one$. 
     \item $(N x)^- \equiv \left(∀X:[\nat], !(∀y, X\,y ⊸ X\,(s\,y)) ⊸ !(X\,0 ⊸ X\,x)\right)^- 
          \equiv ∀α, (α → α) → α →α \equiv \nat$, 
  \end{itemize}
\end{example}

\section{The proof system}

Sequents are of the form $Γ;Δ ⊢ \faded{t:}P$ where $Γ$ is a context (see 
definition \ref{DefContext}), $Δ$ is an unordered set of 
assignments of the form $\faded{x : }Q$ where $t$ is a first-order
term, $x$ a first-order variable and $P$ and $Q$ are formulas. 
Our proof system has two parameters:
  \begin{itemize}
    \item A well-formed typing context $\Sign$ of types of functions we want to implement. 
  In this paper, we use the set 
  $$  \begin{array}{lcclr} 
    \Sign & = & \{ & 0 : \nat, \succe : \nat → \nat, \pred : \nat → \nat, 
                     \mult : \nat → \nat → \nat, & \\
          &   &    & \minus : \nat → \nat → \nat, \plus : \nat → \nat → \nat,& \\
          &   &    & \sumb : (\nat → \nat) → \nat → \nat, \prodb : (\nat → \nat) → \nat → \nat  & \}. 
    \end{array}$$
    \item A set $\H$ of equational formulas of the form $∀ x_1 :τ₁,...,∀x_n:τ_n,t₁ =_τ t₂$ well-typed in 
          $\Sign$. In this paper, we take $\H$ to be the intersection of all sets $T$ of formulas of 
          this form such that $\H_0 ⊧_Σ T$ where $\H_0$ is the set below. 
\begin{small}
  \[
\begin{array}{cccccc}
\H_0 ⊧\{ & & 0 & =_\nat & Λα.λf:α→α.x:α.x & ,\\
 & \forall n:\nat, & s\,n & =_\nat & Λα.λf:α→α.x:α.n\,α\,f (f\, x) & ,\\
 & \forall x\, y:\nat, & \plus\,\, x\,(s\,\, y) & =_\nat & s\,\,(plus\,\, x\,\, y) & ,\\
 & \forall x:\nat, & plus\,\, x\,\,0 & =_\nat & x & ,\\
 & \forall x\, y:\nat, & \mult\,\, x\,\,(s\,\, y) & =_\nat & plus\,\, x\,\,(mult\,\, x\,\, y) & ,\\
 & \forall x:\nat, & \mult\,\, x\,\,0 & =_\nat & 0 & ,\\
 & \forall x:\nat, & \pred\,\,(s\,\, x) & =_\nat & x & ,\\
 &  & pred\,\,0 & =_\nat & 0 & ,\\
 & \forall x\, y:\nat, & \minus\,\, x\,\,(s\,\, y) & =_\nat & \pred\,\,(\minus\,\, x\,\, y) & ,\\
 & \forall x:\nat, & \minus\,\, x\,\,0 & =_\nat & x & ,\\
 & \forall x:\nat,\forall f:\nat\rightarrow\nat, & \sumb\, f\,\,(s\,\, x) & =_\nat & \plus\,\,(\sumb\,\, f\,\, x)\,(f\,\, x) & ,\\
 & \forall f:\nat\rightarrow\nat, & \sumb\,\, f\,\,0 & =_\nat & 0 & ,\\
 & \forall x:\nat,\forall f:\nat\rightarrow\nat, & \prodb\,\, f\,\,(s\,\, x) & =_\nat & \mult\,\,(\prodb\,\, f\,\, x)\,\,(f\,\, x) & ,\\
 & \forall f:\nat\rightarrow\nat, & \prodb\,\, f\,\,0 & =_\nat & s\,\,0 & \}.\end{array}\] 
\end{small}
  \end{itemize}

\begin{fmpage}{0.95\textwidth}
\begin{prooftree}
    \AxiomC{$Σ,Γ \ok P : \Prop$}
    \LeftLabel{\sc \small Axiom}
    \UnaryInfC{$Γ; \faded{x :}P ⊢ \faded{x :} P$}

    \AxiomC{$Γ ; Δ ⊢ \faded{t :} Q$}
    \AxiomC{$Σ, Γ  \ok P : \Prop$}
    \LeftLabel{$x \not∈ Δ$}
    \RightLabel{\sc \small Weakening}
    \BinaryInfC{$Γ ; Δ, \faded{x :} P ⊢ \faded{t :}Q$}
    \noLine
    \BinaryInfC{}
\end{prooftree}\begin{prooftree}
    \AxiomC{$Γ ; Δ₁ ⊢\faded{t₁ :} P ⊸ Q$}
    \AxiomC{$Γ ; Δ₂ ⊢\faded{t₂ :} P$}
    \RightLabel{\sc \small Application}
    \BinaryInfC{$Γ ; Δ₁,Δ₂ ⊢ \faded{(t₁\,t₂):} Q$}

    \AxiomC{$Γ ; Δ, \faded{x :} P ⊢ \faded{t :} Q$}
    \RightLabel{\sc \small Abstraction}
    \UnaryInfC{$Γ ; Δ ⊢ \faded{λx:P^-.t:} P ⊸ Q$}
    \noLine
    \BinaryInfC{}
\end{prooftree}
\begin{prooftree}
   \AxiomC{$Γ ; Δ₁ ⊢ \faded{t₁ :} !P₁\qquad$ ... $\qquadΓ ; Δ_n ⊢ \faded{t_n :} !P_n$}
    \AxiomC{$Γ ; \faded{x₁ :} P₁, ..., \faded{x_n :} P_n ⊢ t : P$}
    \RightLabel{\sc \small Promotion}
    \BinaryInfC{$Γ ; Δ₁,...,Δ_n ⊢ \faded{t[t₁/x₁,...,t_n/x_n] :}!P$}
\end{prooftree}
\begin{prooftree}
    \AxiomC{$Γ ; Δ, \faded{x :}!P,\faded{x :}!P ⊢ \faded{t :} Q$}
    \RightLabel{\sc \small Contraction}
    \UnaryInfC{$Γ ; Δ,\faded{x :}!P ⊢ \faded{t :} Q$}
    
    \AxiomC{$Γ, α : \Type ; Δ ⊢ \faded{t :} P$}
    \RightLabel{\sc \small$∀_α$-Intro}
    \UnaryInfC{$Γ ; Δ ⊢ \faded{t :} ∀ α,P$}
    
    \noLine
    \BinaryInfC{}
\end{prooftree}\begin{prooftree}
    \AxiomC{$Γ, x:τ ; Δ ⊢ \faded{t :} P$}
    \RightLabel{\sc \small $∀_1$-Intro}
    \UnaryInfC{$Γ ; Δ ⊢ \faded{t :} ∀ x:τ,P$}

    \AxiomC{\!$Γ, X:[τ_1,...,τ_n] ; Δ ⊢ \faded{t :} P$}
    \RightLabel{\sc \small$∀_2$-Intro}
    \UnaryInfC{$Γ ; Δ ⊢ \faded{(Λα_X.t) :} ∀ X:[τ_1,...,τ_n],P$}
    \noLine
    \BinaryInfC{}
\end{prooftree}
\begin{prooftree}    
    \AxiomC{$Γ ; Δ ⊢ \faded{t :} ∀ α,P$}
    \AxiomC{$Σ,Γ \ok τ : \Type$}
    \RightLabel{\sc \small $∀_α$-Elim}
    \BinaryInfC{$Γ ; Δ ⊢ \faded{t :} P[τ/α]$}

    \AxiomC{$Γ ; Δ ⊢ \faded{t :} ∀ x:τ,P $}
    \AxiomC{$Σ,Γ \ok a : τ$}
    \RightLabel{\sc \small $∀_1$-Elim}
    \BinaryInfC{$Γ ; Δ ⊢ \faded{t :} P[a/x]$}
    \noLine
    \BinaryInfC{}
\end{prooftree}
\begin{prooftree}
    \AxiomC{$Γ ; Δ ⊢ \faded{t :} ∀ X:[τ_1,...,τ_n],P$}
    \AxiomC{$Σ,Γ, x_1:τ_1, ..., x_n : τ_n \ok Q : Prop$}
    \RightLabel{\sc \small $∀_2$-Elim}
    \BinaryInfC{$Γ ; Δ ⊢ \faded{(t\, Q^-) :} P[Q/X\,x_1...x_n]$}
\end{prooftree}
\begin{prooftree}
    \AxiomC{$Γ;Δ ⊢ \faded{t :} P₁$}
    \LeftLabel{$\H ⊧_{Σ,Γ} t₁ =_τ t₂$ and  $P₁\reduce{$t₁ =_τ t₂$}P₂$}
    \RightLabel{\sc \small Equality}
    \UnaryInfC{$Γ;Δ ⊢ \faded{t :}P₂$}
\end{prooftree}

\begin{center}
\textbf{The proof system parametrized by $Σ$ and $\H$}
\end{center}
\end{fmpage}

The following lemma gives us the type of proof-terms. 
\begin{lem}
  If $Γ;x₁:P₁,...,x_n:P_n ⊢ t : P$, then $Γ^\star,x₁:P₁^-, ...,x_n:P^-_n \ok t : P^-$.
\end{lem} 

\noindent
And this one tells us that our proof system is well-behaved with respect to our
notion of model.
\begin{lem} \label{adequation_model}(Adequacy lemma)\\\indent
  If $Γ;x₁:P₁, ...,x_n:P_n ⊢ t : P$, then 
  $\H \cup \{P₁, ..., P_n\} ⊧_Γ P$. 
\end{lem}
\begin{proof}
  The proof consists of an induction on the structure of the proof 
  $Γ;x₁:P₁, ...,x_n:P_n ⊢ t : P$ and an intensive use of substitution lemmas. 
\end{proof}

\subsection*{A simple realizability theory}

\begin{defn}
Given a formula $F$ and a term $t$, we can recursively define the formula
written $t ⊩ F$ upon the structure of $F$ in the following way.
\begin{itemize}
\item $t⊩ X\, t_{1}\,...\, t_{n} \equiv X\, t_{1}\,...\, t_{n}\,t$,
\item $t⊩ P ⊸ Q \equiv ∀ x:P^{-},x⊩ P\,⊸\,(t\, x)⊩ Q$,
\item $t⊩∀ X:[τ_{1},...,τ_{n}],P \equiv ∀α_{X},∀ X:[τ_{1},...,τ_{n},α_{X}],t\,α_{X}⊩ P$,
\item $t⊩∀ x:τ,P \equiv ∀ x:τ,t⊩ P$,
\item $t⊩∀α,P \equiv ∀α,t⊩ P$,
\item $t⊩!P \equiv !(t⊩ P)$.
\end{itemize}
\end{defn}

\begin{lem}
For any formula $P$ and any context $Γ$ and any first-order term $t$,
\begin{align*}
\left.\begin{aligned}
Γ  \ok\,  P &: \Prop  \\
Γ^\star \ok\,  t &: P^-
\end{aligned} \right\} ⇒ Γ^- \ok (t ⊩ P) : \Prop
\end{align*}
where $Γ^-$ is obtained from $Γ$ by replacing each occurrence of
``$X:[τ_1, ...,τ_n]$"
by ``$α_X : \Type, X : [τ_1, ..., τ_n, α_X]$" (and $Γ^\star ⊆ Γ^-$ as in 
lemma \ref{lemma_moins}).
\end{lem}

\begin{lem} \label{adequation_realizer} (Adequacy lemma for realizers)\\\indent
If $Γ;x₁:P₁,...,x_n:P_n ⊢ t : P$, then \vspace{-0.60em}
    $$Γ,x₁:P₁^-,...,x_n:P_n^-; x₁ : (x₁ ⊩ P₁), ..., x_n : (x_n ⊩ P_n)  ⊢ t : (t ⊩ P).$$
\end{lem}
\begin{proof}
  It is a consequence of the good ``applicative behavior" of realizability. 
  The result comes easily with an induction on the structure of proof of 
  $Γ;x₁:P₁,...,x_n:P_n ⊢ t : P$.
\end{proof}

\subsection*{Programming with proofs}

\begin{defn}
  Let $D$ be a formula such that $Γ, x : τ \ok D\,x:\Prop$ for some $τ$. 
  We say that $D$ is \emph{data type} of parameter $x$ of type $D^-$
   relatively to a $Γ$-model 
  $ℳ$ if we have : 
  \begin{enumerate}
    \item $ℳ ⊧ ∀ r\,x:D^-, (r ⊩ D) ⊸ r =_τ x$, 
    \item $ℳ ⊧ ∀ x:D^-, x ⊩ D$ (or equivalently the converse $∀ r\,x:D^-, r =_τ x ⊸ (r ⊩ D)$ of 1.) 
  \end{enumerate}
  We simply say that $D\,y$ is a data type in $ℳ$, if $D$ is a data type of parameter $y$
  relatively to $ℳ$ and for any term $t$ such that $Γ \ok t : D^-$, we write $D\,t$ 
  instead of $D[t/y]$. 
\end{defn}

\begin{lem}
  $N\,x$ is a data type in all $Σ$-models. 
\end{lem}
\begin{proof}
  The proof is similar that the one for $FA₂$ in \cite{LambdaCalculusTypesAndModels}.
\end{proof}

\begin{lem}
  If $A\,x$ and $B\,y$ are two data types in a $Γ$-model $ℳ$, so is
    $F\,f \equiv ∀x:A^-,A\,x ⊸ B\,(f\,x)$. 
\end{lem}
\begin{proof} We have to verify the two conditions of the definition. 
  \begin{enumerate}
    \item If $ℳ'$ is a $Γ,r:A^-→B^-,f:A^-→B^-$-model such that $ℳ' ⊧ r ⊩ F\,f$. 
      Since $r ⊩ F\,f \equiv ∀ s\,x, s ⊩ A\, x ⊸ (r\,s) ⊩ B\,(f\,x)$ and 
      by invoking the second condition for $A$ and the first for $B$ we have 
      $ℳ ⊧ ∀ s\,x, s =_{A^-} x ⊸ (r\,s) =_{B^-} (f\,x)$ which is equivalent
      by extensionality to $ℳ ⊧ r =_{A^-→B^-} f$. 
    \item Let $ℳ'$ be a $Γ,f:A^-→B^-$-model, we have to prove that $ℳ' ⊧ f ⊩ F\,f$ 
    or equivalently that $ℳ ⊧∀r\,x:A^-x, r⊩A\,x⊸(f\,r)⊩B\,(f\,x)$.
    But according to the first condition for $A$ it is stronger that 
    $ℳ ⊧∀r\,x:A^-x, r=_{A^-}x⊸(f\,r)⊩B\,(f\,x)$ which is implyed the second  
    condition for $B$.
  \end{enumerate}
\end{proof}

The following theorem state that if we can find a model $ℳ$ satisfying $\H$ (informally it means
that we know our specifications to be implementable), then the program $t$ extracted from the proof 
of a formula stating that a function $f$ is provably total implements this function. 
\begin{thm}
  Let $D₁\,x₁$, ..., $D_n\,x_n$, and $D$ be $n+1$ data types. 
  If $Γ \ok f:D₁^- → ... →D_n^-→D^-$
  If $$Γ; ⊢ t : ∀x₁:D₁^-, ...,x_n:D_n^-, D₁\,x₁ ⊸ ... ⊸ D_n\,x_n ⊸ D\,(f\,x₁\,...\,x_n),$$ then 
  for all $Σ,Γ,f:D₁^- → ... →D_n^-→D^-$-model $ℳ$ such $ℳ ⊧ \H$, $$ℳ ⊧t =_{D₁^- → ... →D_n^-→D^-} f.$$ 
\end{thm}
\begin{proof}
  By lemma \ref{adequation_realizer} we have 
    $Γ;⊢t ⊩D₁\,x₁ ⊸ ... ⊸ D_n\,x_n ⊸ D\,(f\,x₁\,...\,x_n)$
  which is equivalent to 
    \vspace{-0.5em}
    $$Γ;⊢∀r₁\,x₁:D₁^-, ..., ∀r_n\,x_n:D_n^-, r₁ ⊩D₁\,x₁ ⊸ ... ⊸ r_n ⊩ D_n\,x_n ⊸ (t\,r₁\,...\,r_n) ⊩ D\,(f\,x₁\,...\,x_n)$$
    \vspace{-0.5em}
  by lemma \ref{adequation_model} we have
    \vspace{-0.5em}
    $$ℳ ⊧∀r₁\,x₁:D₁^-, ..., ∀r_n\,x_n:D_n^-, r₁ ⊩D₁\,x₁ ⊸ ... ⊸ r_n ⊩ D_n\,x_n ⊸ (t\,r₁\,...\,r_n) ⊩ D\,(f\,x₁\,...\,x_n)$$
    \vspace{-0.5em}
  but since every one is a data type we obtain 
    \vspace{-0.5em}
    $$ℳ ⊧∀r₁\,x₁:D₁^-, ..., ∀r_n\,x_n:D_n^-, r₁ =_{D₁^-}x₁ ⊸ ... ⊸ r_n =_{D_n^-} x_n ⊸ (t\,r₁\,...\,r_n) =_{D^-} (f\,x₁\,...\,x_n)$$
    \vspace{-0.5em}
  which is equivalent to
    $ℳ ⊧t =_{D₁^- → ... →D_n^-→D^-} f.$
\end{proof}

\section{Elementary Time Characterisation}

\subsection*{Correctness}

We describe here how we can bring our system back toward
Elementary Affine Logic in order to prove that extracted programs are
elementary bounded. 
In this section, we will consider the grammar of second-order
elementary logic which is basically a linear version of system $\mathcal{F}$
types. 
\begin{equation*}
τ,σ,...  \quad := \quad α
  \quad|\quad ∀α,τ
  \quad|\quad σ ⊸ τ \quad | \quad !τ\end{equation*}

\begin{defn}
Given a formula $F$, we define the type $F^∘$ recursively 
built from $F$ in the following way.
\begin{equation*}
\begin{array}{lcccccr}
(X\, t_1\,...\, t_n)^∘ = α_X &
(A ⊸ B)^∘ = A^∘ ⊸ B^∘ &
(∀α,F)^∘ = F^∘ &
(∀ x:α,F)^∘ = F^∘  &
(!F)^∘ = !F^∘  \\
\end{array}
\end{equation*}
\begin{equation*}
(∀ X:[τ_1,...,τ_n],F)^∘ = ∀α_X,F^∘.
\end{equation*}
\end{defn}

We map the rules of our system by removing first-order with our map $⋅↦⋅^∘$, the rules of 
equality, introduction and elimination for first-order $∀$ and type $∀$ then become
trivial. We also erase some type information on typed terms 
in order to obtain the following \textit{à la} church type system which is known as 
\textit{elementary affine logic}.
\begin{fmpage}{0.95\textwidth}
\begin{prooftree}
    \AxiomC{}
    \RightLabel{\sc \small Axiom}
    \UnaryInfC{$ x : τ \ell x : τ$}

    \AxiomC{$Δ \ell t : σ$}
    \RightLabel{\sc \small Weakening}
    \UnaryInfC{$Δ, x : τ \ell t : σ$}
    
    \AxiomC{$Δ, x:!σ, x:!σ \ell t : τ$}
    \RightLabel{\sc \small Contraction}
    \UnaryInfC{$Δ, x:!σ \ell t : τ$}

    \noLine
    \TrinaryInfC{}
\end{prooftree}\begin{prooftree}
    
    \AxiomC{$Δ₁ ⊢ t_1 : !τ₁$ \qquad ... \qquad $Δ_n ⊢ t_n : !τ_n$}
    \AxiomC{$ x₁ : τ₁,...,x_n : τ_n \ell t : τ$}
    \RightLabel{\sc \small Promotion}
    \BinaryInfC{$Δ₁,...,Δ_n \ell t[t₁/x_1,...,t_n/x_n] : !τ$}

\end{prooftree}\begin{prooftree}

    \AxiomC{$Δ₁ \ell s : τ ⊸ σ$}
    \AxiomC{$Δ₂ \ell t : τ$}
    \RightLabel{\sc \small Application}
    \BinaryInfC{$Δ₁,Δ₂ \ell(s\,\,t) : σ$}

    \AxiomC{$Δ, x : σ \ell t : τ $}
    \RightLabel{\sc \small Abstraction}
    \UnaryInfC{$Δ \ell (λx:σ.t) : σ ⊸ τ$}

    \noLine
    \BinaryInfC{}
\end{prooftree}\begin{prooftree}
    \AxiomC{$Δ \ell t : τ$}
    \LeftLabel{$α\not\in Δ$}
    \RightLabel{\sc \small$∀$-Intro}
    \UnaryInfC{$Δ \ell t : ∀ α,τ$}

    \AxiomC{$Δ \ell t : ∀ α,τ$}
    \RightLabel{\sc \small $∀$-Elim}
    \UnaryInfC{$Δ \ell t : τ[σ/α]$}

    \noLine
    \BinaryInfC{}
\end{prooftree}
\begin{center}
\textbf{Elementary Affine Logic}
\end{center}
\end{fmpage}
\vspace{0.5em}
We use this translation from our type system to elementary affine logic to
obtain the following lemma. 
\begin{lem} 
  If $Γ;Δ \ok t : P$, then $Δ^∘ \ell \overline{t} : P^∘$ where $\overline{t}$ is
  the pure term obtained by removing type information from $t$ and $Δ^∘$ is
  obtained by sending $x:P$ to $x:P^∘$. 
\end{lem}

The data type $N\,x$ representing integers is sent to 
  $(N\,x)^∘ = ∀ α,!(α ⊸ α) ⊸ !(α ⊸ α)$ (denoted $N^∘$). 

\begin{defn}
  We say that a program $t ∈ \Pgr$ represent
  a (set-theoretical) total function $f$ 
  if for all integers $m₁$, ..., $m_n$, the term $(t\,\lceil m₁\rceil\,...\,\lceil m_n\rceil)$
  may be normalized to the church numeral $\lceil f(m₁,...,m_n)\rceil$. 
  We say that $t ∈ \elem$ if it represents a total function $f$
  belonging to the set of elementary computable functions (where
  $\lceil m\rceil$ is the $m$-th Church integer).
\end{defn}

The following lemma is a bit of a folklore result. The closest reference would
be the appendix of \cite{LLL}.

\begin{lem}
 If $\ell t : !^{k₁}N^∘ ⊸ ... ⊸ !^{k_n}N^∘ ⊸ !^k N^∘$ then 
  $t ∈ \elem$.
\end{lem}
\begin{proof} (very rough sketch)
  You can bring the normalization of $(t\,\lceil m₁\rceil\,...\,\lceil
  m_n\rceil)$ back to 
  the normalization of a proof net corresponding to the proof tree that 
  $\ell (t\,\lceil m₁\rceil\,...\,\lceil m_n\rceil) : !^k N^∘$. Promotion rules 
  are represented as boxes in the proof net. These boxes stratify the proof net in 
  the sense that we can define the \emph{depth} of a node to be the number of boxes containing
  this node. And the depth of the net is the maximal depth of its nodes. If $N$ is the size
  of the proof net, then there is a clever strategy to eliminate all cuts at a given depth
  (without changing the depth) by multiplying the size of the net by at most $2^N$. We 
  therefore obtain the exponential tower by iterating this process for each depth. 
\end{proof}

Finally by combining the last two lemmas, we prove the desired correctedness
theorem.
\begin{thm}
  If we have
  $$Γ, f: \nat → ... → \nat; ⊢ t : ∀x₁:\nat...∀x_n:\nat, !^{k_1}N\,x_1 ⊸ ... ⊸
!^{k_n}N\,x_k ⊸ !^k N (f\,x_1\,...\,x_n)$$
  then $\overline{t} ∈ \elem$ where $\overline{t}$ is the untyped λ-term
  obtained by erasing type information from $t$.
\end{thm}

\subsection*{Completeness}

In this section we give two proofs of the fact that all elementary
recursive functions may be extracted from a proof of totality. 

In order to ease the reading on paper, we omit term annotations ( the ``$x
:~$" in $Δ$ and ``$t :$" on the right-hand side of the symbol $⊢$) since, given
a proof tree, theses decorations are unique up to renaming of variables. We also
allow ourselves to let the typing context $Γ$ and proofs of the typing sequents
$\ok$ implicit. Theses three derivable rules will be very useful in the following.
\begin{lem}
   These rules are derivable:
    \begin{prooftree}
        \AxiomC{$Δ ⊢ A$}
        \UnaryInfC{$!Δ ⊢ !A$}
        
        \AxiomC{$Δ,A,B ⊢ C$}
        \UnaryInfC{$Δ,A⊗B ⊢ C$}

        \AxiomC{$Δ₁ ⊢ A$}
        \AxiomC{$Δ₂ ⊢ B$}
        \BinaryInfC{$Δ₁,Δ₂ ⊢ A⊗B$}  
        \noLine
        \TrinaryInfC{}
    \end{prooftree}
\end{lem}

\subsubsection*{First proof of completeness: using the completeness of EAL}

The following theorem gives us a link between typable terms in $ELL$ and
provably total functions in our system. And if we admit the completeness 
of EAL, it gives us directly that all elementary recursive functions may 
be extracted from a proof of totality. 

\begin{thm}
  Let $t$ such that 
    $\ell t : \nat ⊸ ... ⊸ \nat ⊸ !^k \nat$, then 
  $$⊢ ∀ x₁ ... x_n, N\,x₁ ⊸ ... ⊸ N\,x_n ⊸ !^{k+1} N\,(t\, x₁\,...\,x_n).$$
\end{thm}
\begin{proof}
  Let $N$ be the formula $∀ X, !(X ⊸ X) ⊸ !(X ⊸ X)$. 
  We have a natural embedding of EAL in our system by translating type 
  variables to second-order variables. Therefore, we have 
  $⊢ t : N ⊸ ... ⊸ N ⊸ !^k N$ and then $ ⊢ (t ⊩ N ⊸ ... ⊸ N ⊸ !^k N)$ (*). 
  We are going to need the two simple lemmas below:
  \begin{enumerate}
    \item We have $⊢ ∀ r, (r ⊩ N) ⊸ N (r\,\nat\,s\,0)$. \\\noindent The idea of the  
    proof is that $(r ⊩ N)$ is equal to 
    $$∀α, ∀X:[α], ∀ f: α, !(∀ y, X\,y ⊸ X\,(f\,y)) ⊸ 
      !(∀ z, X\,z ⊸ X\,(r\,α\,f\,z))$$
     and by taking $α = \nat$, $y = s$ and $z = 0$, we obtain
        $N (r\,\nat\,s\,0)$.
    \item And we have $⊢ ∀ r, N r ⊸ !(r ⊩ N)$. \\\noindent
      Let $H$ be $!(∀y, y ⊩ N ⊸ (s y) ⊩ N) ⊸ !(0 ⊩ N ⊸ r ⊩ N)$.
    \begin{small}
      \begin{prooftree}
        \AxiomC{}
        \UnaryInfC{$0 ⊩ N ⊸ r ⊩ N ⊢ (0 ⊩ N) ⊸  (r ⊩ N)$}

        \AxiomC{$\vdots$}
        \noLine 
        \UnaryInfC{$π₁$}
        \UnaryInfC{$ ⊢ 0 ⊩ N$}
        \BinaryInfC{$0 ⊩ N ⊸ r ⊩ N ⊢ r ⊩ N$}
        \UnaryInfC{$  !(0 ⊩ N ⊸ r ⊩ N) ⊢ ! (r ⊩ N)$}
        \UnaryInfC{$ ⊢ !(0 ⊩ N ⊸ r ⊩ N) ⊸ ! (r ⊩ N)$} % A ⊸ B

        \AxiomC{}
        \UnaryInfC{$ N\,r ⊢ N\,r$}
        \UnaryInfC{$ N\,r ⊢ H$} % C ⊸ A 

        \AxiomC{$\vdots$}
        \noLine 
        \UnaryInfC{$π₂$}
        \UnaryInfC{$ ⊢ ∀y, y ⊩ N ⊸ (s y) ⊩ N)$}
        \UnaryInfC{$ ⊢ !(∀y, y ⊩ N ⊸ (s y) ⊩ N))$}
               
        \BinaryInfC{$N\,r ⊢ !(0 ⊩ N ⊸ r ⊩ N)$} % A
        \BinaryInfC{$N\,r ⊢ !(r ⊩ N)$} % B
        \doubleLine
        \UnaryInfC{$⊢ ∀ r, N\,r ⊸ !(r ⊩ N)$}
      \end{prooftree}
    \end{small}
     where $π₁$ and $π₂$ use the rule {\sc Equality}
     with 
        $$\H ⊧_{α:\Type,f:α→α,z:α} (0\, α\, f\, z) =_α z \text{ and }
        \H ⊧_{y:\nat,α:\Type,f:α→α,z:α} (s\,y\, α\, f\, z) =_α (y\, α\, f\, (f\,z)).$$
  \end{enumerate}
  Now to prove the sequent
  $⊢ ∀ x₁ ... x_n, N\,x₁ ⊸ ... ⊸ N\,x_n ⊸ !^{k+1} N\,(t\, x₁\,...\,x_n)$, 
  it is enough to find a proof of 
  $⊢ ∀ x₁ ... x_n, !N\,x₁ ⊸ ... ⊸ !N\,x_n ⊸ !^{k} N\,(t\, x₁\,...\,x_n)$ (using
  the {\sc Promotion} rule). By invoking 2, we just have to prove that
  $⊢ ∀ x₁ ... x_n,  (x₁ ⊩ N) ⊸ ... ⊸ (x_n ⊩ N) ⊸ !^{k} N\,(t\, x₁\,...\,x_n)$
  and then by invoking 1, we have to prove
  $⊢ ∀ x₁ ... x_n,  (x₁ ⊩ N) ⊸ ... ⊸ (x_n ⊩ N) ⊸ (t\, x₁\,...\,x_n) ⊩ !^k N$
  which is equivalent to (*).

\end{proof}

\subsubsection*{Second proof of completeness : encoding Kalmar's functions}

The characterization due to Kalmar 
\cite{Rose} states that elementary recursive functions is the smallest class of
functions containing some base functions (constants, projections, addition, multiplication 
and subtraction) and stable by a composition scheme, by
bounded sum and bounded product. In the remaining of the document, we will show
how we can implement this functions and these schemes in our system.

\begin{itemize}
\item It is very easy to find a proof of $⊢ N\,0$ and a proof 
  $⊢∀x,N x ⊸ N (s\,x)$. We can obtain a proof $⊢ N\,(s\,0)$ by composing them.

\item The following proof gives us the addition (in order to make it fit we cut
it in two bits, and the $\vdots$ mean the proof can be easily completed). We use 
``$x + y$" as a notation for the term $(\plus\,x\,y)$.
\begin{small}
\begin{prooftree}
  \AxiomC{$π$}
  %\UnaryInfC{$N\,x, !F ⊢ !(X\,y ⊸ X (x+y))$}
  \AxiomC{$\vdots$}
  \UnaryInfC{$N\,y, !F ⊢ !(X\,0 ⊸ X y)$}
  \AxiomC{$\vdots$}
  \UnaryInfC{$X\,y ⊸ X (x+y), X\,0 ⊸ X y ⊢ X\,0 ⊸ X
(x+y)$}
  \TrinaryInfC{$N\,x,N\,y,!F, !F ⊢ !(X\,0 ⊸ X(s\,x))$}
  \doubleLine
  \UnaryInfC{$⊢∀x\,y:\nat,N\,x ⊸ N\,y ⊸ N (x+y)$}
\end{prooftree}
 \begin{prooftree}
  \AxiomC{}
  \UnaryInfC{$N\,x ⊢ N\,x$}
  \UnaryInfC{$N\,x ⊢ !(∀z, X(z+y) ⊸ X((s\,z)+y)) ⊸!(X\,(0 + y) ⊸ X
(x+y))$}
  \UnaryInfC{$N\,x ⊢ !(∀z, X(z+y) ⊸ X((s\,z)+y)) ⊸!(X\,y ⊸ X (x+y))$\hspace{-0.65em}}
  
  \AxiomC{$\vdots$}
  \UnaryInfC{$!F ⊢ !(∀z, X(z+y) ⊸ X(s\,(z+y)))$}
  \UnaryInfC{$!F ⊢ !(∀z, X(z+y) ⊸ X((s\,z)+y))$}
  \BinaryInfC{$N\,x, !F ⊢ !(X\,y ⊸ X (x+y))$}
  \noLine
  \UnaryInfC{$π$}
  \end{prooftree}
\end{small}
  Note that we have used in the left branch the {\sc Equality} rule with 
    $\H ⊧ ∀x\,y, (s\,x) + y = s\,(x+y)$
  and 
    $\H ⊧ ∀y, 0 + y = y$.
  We extract the usual $λ$-term for addition 
    $λn\,m:\nat.Λα.λf:α→α.λx:α.n\,f\,(m\,f\,x)$.
\item 
By iterating the addition, it is very easy to find a proof of
$∀x\,y:\nat,N\,x⊸N\,y⊸!N\,(\mult\,x\,y)$. Alas in order to build the scheme of
bounded product in the following, we will need to find a proof of
$∀x\,y:\nat,N\,x⊸N\,y⊸ N\,(\mult\,x\,y)$. The proof has been found and checked
using a proof assistant based on our system, but it is too big to fit in there.
       The $λ$-term extracted from this proof
      is 
$\lambda n\, m:\nat.\Lambda\alpha.\lambda
f:\alpha\rightarrow\alpha.n\,\alpha\,(m\,(\alpha\rightarrow\alpha)\,(\lambda
g:\alpha\rightarrow\alpha.\lambda x:\alpha.f\,(g\, x)))\,(\lambda x:\alpha.x)$.

\item We can implement the predecessor function by proving 
    $⊢ ∀ x, N\,x⊸N\,(\pred x)$. The proof is not so easy: you have to
    instantiate a second-order quantifier with $x↦(X p(x) ⊸ X x) ⊗ X p(x)$. 
    It corresponds to a very standard technique for implementing the 
    predecessor of $n$ in $λ$-calculus: we iterate the function $(a,b) ↦ (a+1,a)$
    $n$ times on $(0,0)$ and then we use the second projection to retrieve $n-1$.

\item Then it is easy to implement the subtraction by proving
     $⊢∀x\,y,N\,x⊸N\,y⊸!N\,(\minus\,x\,y)$ with the induction principle $N\,y$.

\item The following proof is called coercion (in \cite{DanosJoinet}), it will 
    allow us to replace occurences of $N\,x$ at a negative position by $!N\,x$. 
    Let $H$ be the formula $∀y,N\,y ⊸ N\,(s\,y)$.
    \begin{small}
    \begin{prooftree}
      \AxiomC{}
      \UnaryInfC{$ N0 ⊸ N\,x ⊢ N0 ⊸ N\,x$}
      \AxiomC{\textit{proof for zero}}
      \UnaryInfC{$ ⊢ N 0$}
      \BinaryInfC{$ N0 ⊸ N\,x ⊢ N\,x$}
      \UnaryInfC{$ !(N0 ⊸ N\,x) ⊢ !N\,x$}
      \UnaryInfC{$⊢ !(N0 ⊸ N\,x) ⊸ !N\,x$}

      \AxiomC{}
      \UnaryInfC{$N\,x ⊢ N\,x$}
      \UnaryInfC{$N\,x ⊢ !H ⊸ !(N0 ⊸ N\,x)$}

      \AxiomC{\textit{proof for successor}}
      \UnaryInfC{$⊢ H$}
      \UnaryInfC{$⊢ !H$}

      \BinaryInfC{$N\,x ⊢ !(N0 ⊸ N\,x)$}
      \BinaryInfC{$N\,x ⊢ !N\,x$}
      \doubleLine
      \UnaryInfC{$⊢∀x,N\,x ⊸ !N\,x$}
    \end{prooftree}
    \end{small}
      Using this we can now bring every proof of totality 
      $$ ⊢ ∀ x₁,...,x_n, !^{k₁}N\,x₁ ⊸ ... ⊸ !^{k_n}N\,x_n ⊸ !^kN\,(f\,x₁\,...\,x_n)$$
    to a ``normal form" 
      $$ ⊢ ∀ x₁,...,x_n, N\,x₁ ⊸ ... ⊸ N\,x_n ⊸ !^kN\,(f\,x₁\,...\,x_n).$$
  
 \item The composition scheme is implemented by the following proof 
      (where $s=\sum_{i=1}^q k_i$ and where $A^{(p)}$ means $A$ is duplicated
$p$ times).
   \begin{small}
    \begin{prooftree}
      \AxiomC{$
          \stackrel{\text{\textit{\normalsize proof for $g₁$}}}{
      \overline{N\,x₁,...,N\,x_q⊢!^{k_1}N\,(g₁\,x₁\,...\,x_q)}}
        \quad     ... \quad
          \stackrel{\text{\textit{\normalsize proof for $g_p$}}}{
            \overline{N\,x₁,... ,N\,x_q⊢!^{k_q}N\,(g₁\,x₁\,...\,x_q)}}$
      }
      \AxiomC{$π$}
      \BinaryInfC{$ (N\,x₁)^{(p)}, ..., (N\,x_q)^{(p)} ⊢
                    !^{s+k} N\,(f\,(g₁\,x₁...\,x_q)...(g_p\,x₁...\,x_q))$}
      \UnaryInfC{$(!N\,x₁)^{(p)}, ..., (!N\,x_q)^{(p)} ⊢
                    !^{s+k+1}N\,(f\,(g₁\,x₁...\,x_q)...(g_p\,x₁...\,x_q))$}
      \doubleLine
      \UnaryInfC{$!N\,x₁, ..., !N\,x_q ⊢
                    !^{s+k+1}N\,(f\,(g₁\,x₁...\,x_q)...(g_p\,x₁...\,x_q))$}
      \doubleLine
      \UnaryInfC{$N\,x₁, ..., N\,x_q ⊢
                    !^{s+k+1}N\,(f\,(g₁\,x₁...\,x_q)...(g_p\,x₁...\,x_q))$}
      \doubleLine
      \UnaryInfC{$⊢∀x₁...x_n, N\,x₁ ⊸ ... ⊸ N\,x_q ⊸
                    !^{s+k+1}N\,(f\,(g₁\,x₁...\,x_q)...(g_p\,x₁...\,x_q))$}
    \end{prooftree}
    \begin{prooftree}
      \AxiomC{\textit{proof for $f$}}
      \UnaryInfC{$N\,(g₁\,x₁\,...\,x_q),..., N\,(g_p\,x₁\,...\,x_q) ⊢
                  !^{k}N\,(f\,(g₁\,x₁...\,x_q)...(g_p\,x₁...\,x_q))$}
      \UnaryInfC{$!^{s}N\,(g₁\,x₁\,...\,x_q),..., !^{s} N\,(g_p\,x₁\,...\,x_q) ⊢
                  !^{s+k}N\,(f\,(g₁\,x₁...\,x_q)...(g_p\,x₁...\,x_q))$}
      \UnaryInfC{$!^{s}N\,(g₁\,x₁\,...\,x_q),..., !^{s} N\,(g_p\,x₁\,...\,x_q) ⊢
                  !^{s+k}N\,(f\,(g₁\,x₁...\,x_q)...(g_p\,x₁...\,x_q))$}
      \UnaryInfC{$⊢!^{k_1}N\,(g₁\,x₁\,...\,x_q) ⊸ ⋅⋅⋅ ⊸ !^{k_p} N\,(g_p\,x₁\,...\,x_q) ⊸
                  !^{s+k}N\,(f\,(g₁\,x₁...\,x_q)...(g_p\,x₁...\,x_q))$}
      \noLine
      \UnaryInfC{$π$}
    \end{prooftree}
    \end{small}

\item Finally, the bounded sum is implemented by the following proof of 
$ !!(∀y,N\,y ⊸ !^kN\,(f\,y)) ⊸ ∀n, N\,n ⊸  !^{k+2}N (\sumb\,f\, n).$
  The key idea in this proof is to use the induction principle of $N\,n$ with
  the predicate $x↦N\,x ⊗ !^kN\, (\sumb\,f\,x)$.
  Let $H$ be the formula $∀y,N\,y ⊸ !^kN\,(f\,y)$ and $K₁$ be the formula 
    $$∀ y, !(N\,y ⊗ !^kN (\sumb\,f y)) ⊸ !(N\,(s\,y) ⊗ !^kN
(\sumb\,f\,(s\,y)))$$
  and $K₂$ the formula $!(N 0 ⊗ !^kN (\sumb\,f\, 0)) ⊸ !(N\,n ⊗ !^kN (\sumb\, f\, n))$.
  \begin{small}
  \begin{prooftree}
    \AxiomC{$N\,n ⊢ N\,n$}
    \UnaryInfC{$N\,n ⊢ !K₁ ⊸ !K₂$}
   
    \AxiomC{π}
    \BinaryInfC{$N\,n, !!H ⊢ !K₂$}

    \AxiomC{\textit{proof for zero}}
    \UnaryInfC{$⊢ N\,0$}
    \AxiomC{\textit{proof for zero}}
    \UnaryInfC{$⊢ N\,0$}
    \doubleLine
    \UnaryInfC{$⊢ !^kN\,0$}
    \BinaryInfC{$⊢ N 0 ⊗ !^kN 0$}
    \UnaryInfC{$⊢ !(N 0 ⊗ !^kN 0)$}
    \UnaryInfC{$⊢ !(N 0 ⊗ !^kN (\sumb\,f\, 0))$}

    \AxiomC{}
    \UnaryInfC{$!^kN (\sumb\, f\, n) ⊢ !^kN (\sumb\, f\, n)$}
    \UnaryInfC{$N n, !^k N (\sumb\, f\, n) ⊢ !^kN (\sumb\, f\, n)$}
    \UnaryInfC{$N n ⊗ !^kN (\sumb\, f\, n) ⊢ !^kN (\sumb\, f\, n)$}
    \UnaryInfC{\hspace{-0.4em}$!(N n ⊗ !^kN (\sumb\, f\, n)) ⊢ !^{k+1}N (\sumb\, f\, n)$\hspace{-0.4em}}
    \BinaryInfC{$K₂ ⊢ !^{k+1}N (\sumb\, f\, n)$}
    \BinaryInfC{$!!H, N\,n ⊢ !^{k+2}N (\sumb\, f\, n)$}
    \UnaryInfC{$⊢ !!(∀y,N\,y ⊸ !^kN\,(f\,y)) ⊸ ∀n, N\,n ⊸  !^{k+2}N (\sumb\,f\, n)$}
  \end{prooftree}
  \end{small}
  \begin{small}
  \begin{prooftree}
    \AxiomC{\textit{proof for successor}}
    \UnaryInfC{$N\,y ⊢  N\,(s\,y)$}
    \UnaryInfC{$N\,y, !^kN (\sumb\,f y) ⊢  N\,(s\,y)$}
    \UnaryInfC{$N\,y ⊗ !^kN (\sumb\,f y) ⊢  N\,(s\,y)$}

    \AxiomC{\textit{proof for addition}}
    \UnaryInfC{$N (\sumb\,f y), N (f\,y) ⊢ N ((f\,y) + (\sumb\,f\,y))$}
    \doubleLine
    \UnaryInfC{$!^kN (\sumb\,f y), !^kN (f\,y) ⊢ !^kN ((f\,y) + (\sumb\,f\,y))$}
    \UnaryInfC{$!^kN (\sumb\,f y) ⊢ !^kN (f\,y) ⊸ !^kN ((f\,y) + (\sumb\,f\,y))$}
    
    \AxiomC{$\vdots$}
    \UnaryInfC{$H, N\,y ⊢ !^kN (f\,y)$}
    \BinaryInfC{$H, N\,y , !^kN (\sumb\,f y) ⊢ !^kN ((f\,y) + (\sumb\,f\,y))$}
    \UnaryInfC{$H, N\,y , !^kN (\sumb\,f y) ⊢ !^kN (\sumb\,f\,(s\,y))$}
    \UnaryInfC{$H, N\,y ⊗ !^kN (\sumb\,f y) ⊢ !^kN (\sumb\,f\,(s\,y))$}
    \BinaryInfC{$H, N\,y ⊗ !^kN (\sumb\,f y), N\,y ⊗ !^kN (\sumb\,f y) ⊢ N\,(s\,y) ⊗ !^kN (\sumb\,f\,(s\,y))$}
    \UnaryInfC{$!H, !(N\,y ⊗ !^kN (\sumb\,f y)), !(N\,y ⊗ !^kN (\sumb\,f y)) ⊢ !(N\,(s\,y) ⊗ !^kN (\sumb\,f\,(s\,y)))$}
    \UnaryInfC{$!H, !(N\,y ⊗ !^kN (\sumb\,f y)) ⊢ !(N\,(s\,y) ⊗ !^kN (\sumb\,f\,(s\,y)))$}
    \doubleLine
    \UnaryInfC{$!H ⊢ K₁$}
    \UnaryInfC{$!!H ⊢ !K₁$}
    \noLine
    \UnaryInfC{π}
   \end{prooftree}
  \end{small}
  and we obtain the bounded product by replacing proofs for zeros by proof for
ones and the proof for addition by a proof for multiplication.
\end{itemize}

\bibliographystyle{eptcs} 
\bibliography{bibliography}

\end{document}